%% file: root.tex
\documentclass[letterpaper, 10 pt, conference]{ieeeconf}  

\IEEEoverridecommandlockouts                             
                                                        
\usepackage{amsmath, amssymb, graphicx} 
\usepackage{float}
\usepackage{hyperref}
\usepackage{graphicx}
\usepackage{caption}
\usepackage{subcaption}
\usepackage{cite}
\usepackage{color}
\usepackage{algorithm}
\usepackage{algpseudocode}
\usepackage{soul}

\newtheorem{theorem}{Theorem} 
\newtheorem{lemma}[theorem]{Lemma}     

\newtheorem{definition}[theorem]{Definition}

\newtheorem{assumption}{Assumption}

\title{\LARGE \bf
Distributed Time-Varying Gaussian Regression via Kalman Filtering
}

\author{
    Nicola Taddei${}^\ast$  \and
     Riccardo Maggioni${}^\ast$ \and Jaap Eising \and Giulia De Pasquale \and Florian D\"orfler \thanks{* Authors contributed equally. Nicola Taddei, Riccardo Maggioni, Jaap Eising, Florian D\"orfler are with the Automatic Control Laboratory, Department of Electrical Engineering and Information Technology, ETH Zurich, Physikstrasse
3 8092 Zurich, Switzerland. (e-mail:\{taddein, rmaggioni, jeising, dorfler\}@ethz.ch).
Giulia De Pasquale is with with the Control Systems Group, Department of Electrical Engineering, TU Eindhoven, Flux
Groene Loper 19
5612 AP Eindhoven, The 
Netherlands. (e-mail: g.de.pasquale@tue.nl). This research is supported by the Swiss National Science Foundation under NCCR Automation.}
}

\begin{document}

\maketitle

\thanks{This paper has been accepted for presentation at the 2025 European Control Conference (ECC). The final version will appear in the ECC Proceedings. © EUCA. Personal use of this material is permitted. Redistribution or reuse must comply with EUCA policies.}

\begin{abstract}
We consider the problem of learning time-varying functions in a distributed fashion, where agents collect local information to collaboratively achieve a shared estimate.
This task is particularly relevant in control applications, whenever real-time and robust estimation of dynamic cost/reward functions in safety critical settings has to be performed. In this paper, we adopt a finite-dimensional approximation of a Gaussian Process, corresponding to a Bayesian linear regression in an appropriate feature space, and propose a new algorithm, DistKP, to track the time-varying coefficients via a distributed Kalman filter. The proposed method works for arbitrary kernels and under weaker assumptions on the time-evolution of the function to learn compared to the literature.
We validate our results using a simulation example in which a fleet of Unmanned Aerial Vehicles (UAVs) learns a dynamically changing wind field.
\end{abstract}

\section{INTRODUCTION}

\emph{Problem Description and Motivation.}
Learning and exploiting spatio-temporal functions is key in areas like traffic \cite{10.1145/2030112.2030126}, temperature control \cite{Wang2019ModelingOA}, and optimization \cite{9636520}, where such functions describe phenomena like urban traffic flow, regional temperature, or cost and reward structures. In all of these circumstances, accurate estimates are essential for safe and efficient operation.

In large, dynamic, and safety-critical environments, relying on networked agents is often advantageous, as they can cover wide areas and gather localized data more efficiently than a single unit \cite{Prieto_2024}. The function estimation process in such systems can follow either a centralized or a distributed approach~\cite{10.1007/978-94-017-2376-3_2}. Although centralized learning is well-established \cite{Goodfellow-et-al-2016}, \cite{buskirk2018introduction}, a distributed approach is often preferred in applications where efficiency and robustness are essential.
In such approaches, as long as some agents remain active, the network continues to operate, even in the face of partial failures.
They also enhance privacy since they don't rely on centralized servers, thus reducing the risk of large-scale data breaches. 
For example, consider a road condition estimator used in assisted driving systems \cite{BoschRoadCondition}: in a centralized architecture, a breach of the server could expose the real-time locations of the entire fleet of vehicles, while a fully distributed solution eliminates this risk by storing position data locally.

\emph{Literature Review.}
Gaussian Processes (GPs) \cite{gp_blr} are commonly used for function estimation \cite{li2024safe,pmlr-v37-sui15}, where they outperform the other techniques due to their data-efficiency, interpretability and their ability to model uncertainty, which makes them particularly suited for safety critical settings \cite{cautious_mpc_gp,safe_mpc,li2024safe}. Although very effective, Gaussian Processess are limited in applicability by their inference time that scales cubically with the number of samples.

While several methods for efficient inference are available for the stationary case \cite{gp_blr},\cite{karhunen_loeve}, the time-varying setting remains challenging because of the incremental nature of the dataset \cite{li2024safe}. Arguably some of the best methods that deal with non-stationary functions leverage Sparse GPs \cite{unifying_sparse_gp}, which rely on the use of a subset of the dataset, and Spatio-temporal GPs \cite{spatio_temporal_gp, li2024safe} that, under assumptions on the time-separability of the kernel, rewrite the Gaussian Process as a state-space model and use filtering to perform inference. The authors in \cite{spatio_temporal_variational_gps} use these techniques in combination with variational inference to further increase efficiency.

While these methods are designed for centralized data collection and inference, some works address the problem of learning distributed Gaussian Processes. Particularly interesting for their design simplicity are methods that leverage finite dimensional Gaussian Process approximations and average consensus algorithms \cite{distributed_gp_for_env_model} to cooperatively learn a common estimate. The paper \cite{distributed_informative_planning}, while mainly focusing on stationary functions, introduces a variant to \cite{distributed_gp_for_env_model} that allows time-varying regression by exponentially damping old information. However, since information is discarded irrespectively of the system's configuration, this method leads to complete local memory loss when recent measurements are not available for a certain area, thus forcing the agents to continuously cover the whole domain in order to avoid it. In practice, this means that more agents have to be deployed, making the method less appealing.

\emph{Contribution}. 
We present a new algorithm, DistKP, that enables distributed learning of time-varying functions. This method
\begin{itemize}
    \item overcomes limitations of previous methods in the context of learning-based control, requiring weaker assumptions on the temporal dynamics;
    \item has only two, easily interpretable, additional parameters compared to standard Gaussian Processes;
    \item allows to arbitrarily trade precision with efficiency.
\end{itemize}
We validate the effectiveness of the  algorithm through simulations, by learning the wind distribution in a cumulus cloud model. 

\emph{Paper Organization}.
The remainder of this paper is organized as follows:
Section \ref{sec: Problem formulation} introduces the problem formulation. In Section \ref{sec: gp} we present an overview of Gaussian Processes, while in Section \ref{sec: technical adaptations} we elaborate some technical adaptations to our specific problem. In Section \ref{sec: problem solution}, we detail DistKP. Finally, Section \ref{sec: numerical results} introduces a test case, where we consider a swarm of  UAVs tasked to learn the wind distribution, represented by a function of the spatial coordinates.

\emph{Notation}.
We denote by $\mathbb{R}$ the set of real numbers.
  For a vector $\mathbf{v}\in\mathbb{R}^n$, $\|\mathbf{v}\|$ denotes the Euclidean norm of $\mathbf{v}$. The identity matrix in $\mathbb{R}^{n\times n}$ is denoted by $I_n$, and we omit the subscript if clear from context. A scalar $x$ (vector $\mathbf{x}$) sampled by a  Gaussian distribution with mean $\mu$ (${\boldsymbol \mu}$) and variance $\sigma^2$ (covariance $\Sigma$) is denoted as $x \sim \mathcal{N}(\mu,\sigma^2)$ ($\mathbf{x} \sim \mathcal{N}(\boldsymbol{\mu}, \Sigma)$).   

\section{Problem Formulation}
\label{sec: Problem formulation}

In this work, we consider the problem of learning a spatial-time varying function, using $n$ agents. To be precise, we consider learning a function $f(\cdot,t):  \mathcal{P} \to \mathbb{R}$, where $\mathcal{P} \subset \mathbb{R}^s$ is a bounded subset of $\mathbb{R}^s$. The learned function can be utilized by a controller to optimize the agents' actions. 

Each of the $n$ agents has a state $\mathbf{x}_i\in \mathbb{R}^s$, $i\in \{1,\dots, n\}$ which  follows single integrator dynamics 
\[
    \mathbf{x}_i(t+1) = \mathbf{x}_i(t) + \mathbf{u}_i(t) , \text{  } \| \mathbf{u}_i \| \leq u_{\text{max}}.
\]
In other words, at each time step $t$ the agent $i$ chooses a displacement vector $\mathbf{u}_i(t)$ such that $\| \mathbf{u}_i(t) \| \leq u_{\text{max}}$.  

The agents can communicate with each other within a certain distance $d$. 
To implement this behavior, we define the set of neighbors for agent $i$ as 
\[
\mathcal{N}_i(t) = \{ j \mid \| \mathbf{x}_i(t) - \mathbf{x}_j(t) \| \leq d, \,\,\, j \in \mathcal{N} \setminus \{i\} \},
\]
where $\mathcal{N} = \{ 1, 2, \ldots, n \}$ is the index set of agents and $d$ is the communication range. Thus, the resulting graph is undirected.
In the rest of this manuscript, we will operate under the following assumption.
\begin{assumption}
    The communication network is strongly connected for each time instant.
\end{assumption}
This assumption can be ensured by adding constraints on the distance between agents in the distributed controller. 

At each time step $t$ each agent gets a noisy observation of the learning function $y_i(t) = f(\mathbf{x}_i(t),t) + \nu_i(t)$, where $\nu_i(t)$ denotes zero mean\footnote{Non-zero mean noise can be easily handled by subtracting its mean value from the measurements.} i.i.d. Gaussian noise $\nu_i(t) \sim \mathcal{N}(0, \sigma_\nu^2)$. 

We will model $f$ using a Gaussian Process (GP), since it explicitly models uncertainty and efficiently uses data, thus making Gaussian Processes an ideal solution for embedded systems with low computational budget and safety critical applications. 

\section{Preliminaries on Gaussian Processes}
\label{sec: gp}
We will introduce appropriate modifications to standard Gaussian Processes to allow a distributed implementation and enable their use on time-varying functions.


Gaussian Processes are a powerful, non-parametric approach to regression and classification. We assume that any finite set of observations can be modeled as a multivariate Gaussian distribution, that is, we define a distribution over functions that is fully specified by its mean and covariance.\newline
Consider a function $m: \mathbb{R}^s \rightarrow \mathbb{R}$ called the \emph{mean function} and a function $k: \mathbb{R}^s \times \mathbb{R}^s \rightarrow \mathbb{R}$, positive-definite and symmetric, that models the covariance between measurements taken at two points of the domain. The latter function is referred to as \emph{kernel function} \cite{kernels}. For convenience, in the following we will assume the mean function to be zero, but our results can be easily adapted to the non-zero mean case. The covariance matrix relative to two sets of points $\mathcal{X} = \{ \mathbf{x}_i \}_{i=1}^N$, $\mathcal{Z} = \{ \mathbf{z_j} \}_{j=1}^M$ is defined as 
\[
    K_{\mathcal{X} \mathcal{Z}} = \begin{bmatrix} k(\mathbf{x}_1,\mathbf{z}_1) & \dots & k(\mathbf{x}_1,\mathbf{z}_M) \\ \vdots & \ddots & \vdots \\ k(\mathbf{x}_N,\mathbf{z}_1) & \dots & k(\mathbf{x}_N,\mathbf{z}_M) \end{bmatrix}.
\]
When $\mathcal{Z} = \{ \mathbf{z}_1 \}$, with a slight abuse of notation, we refer to the covariance matrix as $K_{\mathcal{X} \mathbf{z}_1}$.

Given a vector of measurements $\mathbf{y} = [y_1, \, \dots \, y_N]^\top$ corresponding to evaluations of a time-invariant function $g: \mathcal{P}\rightarrow \mathbb{R}$ at the points of the set $\mathcal{X}$ and a point of the domain $\mathbf{x}\in \mathcal{P}$, where we want to make inference, the posterior probability distribution of $g(\mathbf{x})$ is 
   $ g(\mathbf{x}) \sim \mathcal{N}(\mu_{\mathbf{x}}, \sigma_{\mathbf{x}})$,where
\begin{subequations}\label{eq:GP_inference}
\begin{equation}
    \label{eq:GP_mean_inference}
    \mu_{\mathbf{x}} = m(\mathbf{x}) + K_{\mathbf{x} \mathcal{X}} [K_{\mathcal{X} \mathcal{X}} + \sigma_\nu^2 I]^{-1} \boldsymbol{y},
\end{equation}
\begin{equation}
    \label{eq:GP_var_inference}
    \sigma_{\mathbf{x}} = K_{\mathbf{x} \mathbf{x}} - K_{\mathbf{x} \mathcal{X}} [K_{\mathcal{X} \mathcal{X}} + \sigma_\nu^2 I]^{-1} K_{\mathcal{X} \mathbf{x}}.
\end{equation}
\end{subequations}

Given any two domain points $\mathbf{x}, \mathbf{x}' \in \mathbb{R}^s$, some kernels can be rewritten as $k(\mathbf{x}, \mathbf{x}') = \Phi(\mathbf{x})^\top \Phi(\mathbf{x}')$, where $\Phi: \mathbb{R}^s \to \mathbb{R}^E$ is called \emph{feature map} and $\Phi(\mathbf{x})$ is referred to as the \emph{feature vector} of $\mathbf{x}$ \cite{gp_blr}. We use $\Phi(\mathcal{X})$ to refer to the matrix whose columns are the features of the points of $\mathcal{X}$. In this case the functions described by the Gaussian Process are $g(\mathbf{x}) = \boldsymbol{\theta}^\top \Phi(\mathbf{x})$, where $\boldsymbol{\theta}~\in~\mathbb{R}^E$ is a vector of parameters. In this case, it is known that Gaussian Process regression is equivalent to Bayesian Linear Regression (BLR), as mentioned by the following lemma.

\begin{lemma}[Equivalence of GP and BLR  \cite{GP_in_ML}]
\label{thm: GP to BLR}
    Consider a GP with kernel $k(\cdot, \cdot)$ such that $k(\mathbf{x}, \mathbf{x}') = \Phi(\mathbf{x})^\top \Phi(\mathbf{x}')$. Then the GP regression is equivalent to a Bayesian Linear Regression on features given by $\Phi(\cdot)$.
\end{lemma}

When performing Gaussian Process regression the main design choice is that of the kernel. This choice has to be made in accordance to the assumed properties of the unknown function. Some common choices are \emph{Radial Basis Function} (RBF) kernels and \emph{Laplace} kernels. 
\begin{definition}[Radial Basis Function kernel]
   Given a Gaussian Process with covariance function $k:\mathbb{R}^s\times\mathbb{R}^s\to \mathbb{R}$, the Radial Basis Function kernel is defined as
\[
    k_{\rm RBF}(\mathbf{x}, \mathbf{x}') := \exp \left( -\frac{1}{2 \ell^2} \| \mathbf{x} - \mathbf{x}' \|^2 \right),
\]
where \( \ell \) is the length scale hyperparameter. 
\end{definition}
Using a RBF kernel enforces that the estimated function $g \in \mathcal{C}^\infty$.

\begin{definition}[Laplace kernel]
    Given a GP with covariance function $k:\mathbb{R}^s\times\mathbb{R}^s\to \mathbb{R}$, the Laplace kernel is defined as
    \[
    k_{\rm L}(\mathbf{x}, \mathbf{x}') = \exp \left( -\frac{\| \mathbf{x} - \mathbf{x}' \|}{\ell} \right).
    \]
\end{definition}
In contrast to the RBF kernel, the Laplace kernel can be used to learn $\mathcal{C}^0$ functions. Other choices are used when functions are assumed to be periodic, or otherwise structured. 

\section{Technical adaptations}
\label{sec: technical adaptations}
Recall that, with respect to standard Gaussian Processes, our problem requires dealing with 3 additional difficulties:
\begin{enumerate}
    \renewcommand{\theenumi}{\Alph{enumi}}
    \item  A large, and steadily growing, amount of data needs to be processed.
    \item In a time-varying setting the estimate is required to adapt to the changing function and the algorithm needs to be faster then the function dynamics.
    \item Since we deal with a multi-agent system and our main concern is robustness and privacy, we adopt a distributed implementation.
\end{enumerate}
In the following, we address these challenges in order.

\subsection{Nyström approximation}
\label{sec: Nystrom approximation}

When the number of measurements $N$ is large, Gaussian Process inference \eqref{eq:GP_inference} becomes computationally cumbersome due to the computation of the matrix inverse: $\left(K_{\mathcal{X} \mathcal{X}} + \sigma_\nu^2 I\right)^{-1}$, that requires $\mathcal{O}(N^3)$ operations. 

Given that we need to learn a function that changes over time in an online setting, with possibly limited computational resources, efficiency is a key requirement. We adopt a finite-order GP approximation that allows us to arbitrarily trade precision with efficiency through the choice of the hyperparameter $E$ (the order of the approximation). One option to achieve this is to use a truncated Karhunen–Loève (KL) expansion {\cite{karhunen_loeve}}. However, this restricts the applicable kernels to the ones for which a close-form expression for the eigenfunctions exists, i.e. RBF but not Laplace. A kernel agnostic alternative is the Nyström method \cite{nystrom} that enables efficient inference by approximating large covariance matrices under a low-rank assumption.

Recall the standard Gaussian Process regression of Section~\ref{sec: gp}. In addition to $\mathcal{X}$, consider a set of $E$ \emph{representative points} $\mathcal{R}$ such that $E \ll N$ and the complete set of points $\mathcal{D} = \mathcal{R} \cup \mathcal{X}$. In practice $\mathcal{R}$ can be either a subset of the training set or a completely new set of points, for which we don't require any measurement. The eigen-decomposition of $K_{\mathcal{D} \mathcal{D}}$ is then
\[
    K_{\mathcal{D} \mathcal{D}} = \begin{bmatrix} K_{\mathcal{R} \mathcal{R}} & K_{\mathcal{R} \mathcal{X}} \\ K_{\mathcal{X} \mathcal{R}} & K_{\mathcal{X} \mathcal{X}} \end{bmatrix} = U \begin{bmatrix} \Lambda_{\mathcal{R}} & \boldsymbol{0} \\ \boldsymbol{0} & \Lambda_{\mathcal{X}} \end{bmatrix} U^\top,
\]
\[
    U = \begin{bmatrix} U_{11} & U_{12} \\ U_{21} & U_{22} \end{bmatrix},
\]
where $U \in \mathbb{R}^{(E+N) \times (E+N)}$ is an orthonormal matrix and $\Lambda_{\mathcal{R}} \in \mathbb{R}^{E \times E}$, $\Lambda_{\mathcal{X}} \in \mathbb{R}^{N \times N}$ are diagonal matrices.

The Nyström method allows us to efficiently compute the inverse of the covariance matrix $K_{\mathcal{D} \mathcal{D}}$ under the following assumption.
\begin{assumption}[Low-rank of the covariance matrix]
\label{ass:negl_eigenv}
The entries of the matrix $\Lambda_{\mathcal{X}}$ are negligible. Namely,
\[
    K_{\mathcal{D} \mathcal{D}} \approx \begin{bmatrix} U_{11} \\ U_{21} \end{bmatrix} \Lambda_{\mathcal{R}} \begin{bmatrix} U_{11}^\top & U_{21}^\top \end{bmatrix}.
\]
\end{assumption}

Intuitively, this means that the covariance matrix contains redundant information that can be compressed.

\begin{theorem}[Nyström approximation \cite{nystrom}]
\label{theorem:nystrom}
Let Assumption \ref{ass:negl_eigenv}  hold. The covariance of the set $\mathcal{X}$ can be expressed as
\begin{equation}
    \label{eq: Nystrom kernel}
    K_{\mathcal{X} \mathcal{X}} = K_{\mathcal{X} \mathcal{R}} \,\, U_{11} \,\, \Lambda_{\mathcal{R}}^{-1} \,\, U_{11}^\top \,\, K_{\mathcal{R} \mathcal{X}},
\end{equation}
and the matrix expression $K_{\mathbf{x}\mathcal{X}}(K_{\mathcal{X} \mathcal{X}} + \sigma_\nu^2 I)^{-1}$ in \eqref{eq:GP_inference} as
\begin{align*}
    &K_{\mathbf{x}\mathcal{X}} \left( K_{\mathcal{X} \mathcal{X}} + \sigma_\nu^2 \,\, I\right)^{-1} = \\
    &= K_{\mathbf{x}\mathcal{R}} U_{11} \left( U_{11}^\top K_{\mathcal{R} \mathcal{X}} K_{\mathcal{X} \mathcal{R}} U_{11} + \sigma_\nu^2 \,\, \Lambda_{\mathcal{R}} \right)^{-1} U_{11}^\top K_{\mathcal{R} \mathcal{X}}.
\end{align*}
\end{theorem}
This approximation reduces the time complexity of computing $\left(K_{\mathcal{X} \mathcal{X}} + \sigma_\nu^2 I \right)^{-1}$  from $\mathcal{O}(N^3)$ to $\mathcal{O}(E^3)$, thus making inference computationally tractable even for large $N$. 

Approximate inference with the Nyström method in Theorem \ref{theorem:nystrom} is equivalent to a standard GP with the kernel specified in the following lemma.
\begin{lemma}[Nyström induced feature space] 
\label{lemma:Nystrom GP}
Inference with the Nyström method is equivalent to inference with Gaussian Processes with kernel $k(\mathbf{x}, \mathbf{x}') = K_{\mathbf{x} \mathcal{R}} U_{11} \Lambda_{\mathcal{R}}^{-1} U_{11}^\top K_{\mathcal{R} \mathbf{x}'}$ and feature map $\Phi: \mathbb{R}^s \rightarrow \mathbb{R}^E$ 
\begin{equation}
    \Phi(\mathbf{x}) = \Lambda_{\mathcal{R}}^{-\frac{1}{2}} U_{11}^\top K_{\mathcal{R} \mathbf{x} }.
\end{equation}
\end{lemma}
\begin{proof}
    Consider GP mean inference with kernel $k(\mathbf{x}, \mathbf{x}') = \Phi(\mathbf{x})^\top \Phi(\mathbf{x}')$,
    \begin{align*}
        \mu(\mathbf{x}) &= \Phi(\mathbf{x})^\top \Phi(\mathcal{X}) \,\, \left(\Phi(\mathcal{X})^\top \Phi(\mathcal{X}) + \sigma_{\nu}^2 \mathbb{I}\right)^{-1} \,\mathbf{y}.
    \end{align*}
    By applying the \emph{push-through identity} to this expression, we obtain the Nyström approximate inference formula, thus proving the equivalence of the two methods
    \begin{align*}
        \mu(\mathbf{x}) &= \Phi(\mathbf{x})^\top \,\, \left( \Phi(\mathcal{X}) \Phi(\mathcal{X})^\top + \sigma_\nu^2\mathbb{I} \right)^{-1} \Phi(\mathcal {X}) \,\mathbf{y} = \\
        &= K_{\mathbf{x}\mathcal{R}} U_{11} \left( U_{11}^\top K_{\mathcal{R} \mathcal{X}} K_{\mathcal{X} \mathcal{R}} U_{11} + \sigma_\nu^2 \Lambda_{\mathcal{R}} \right)^{-1} U_{11}^\top K_{\mathcal{R} \mathcal{X}} \mathbf{y}.
    \end{align*}
    The proof for the variance is similar.
\end{proof}
In line with Lemma~\ref{thm: GP to BLR}, we see that a Gaussian Process approximated with the Nyström method is equivalent to a BLR on artificial measurements modeled as: 
\begin{equation}\label{eq:Nystrom Phi}
    y_i = \boldsymbol{\theta}^\top \Phi(\mathbf{x}_i) + \nu_i,
\end{equation}
where $\boldsymbol{\theta} \in \mathbb{R}^E$ is the parameter vector.

\subsection{Time-varying function}
Now that we have addressed the first issue, we move our attention to the remaining two. 
As noted in Section \ref{sec: Nystrom approximation}, a time invariant Gaussian Process approximated with the Nyström method corresponds to a BLR in a space of features $\Phi(\mathbf{x})$, that models the function as $g(\mathbf{x}) = \Phi(\mathbf{x})^\top \boldsymbol{\boldsymbol{\theta}}.$

Therefore all the information about the function is stored in a parameter vector $\boldsymbol{\theta}$. In order to learn time-varying functions, $f(\cdot,t):  \mathcal{P} \to \mathbb{R}$, we will assume that the parameter vector $\boldsymbol{\theta}$ is time-varying, that is, 
\[ f(\mathbf{x},t) = \Phi(\mathbf{x})^\top \boldsymbol{\theta}(t).\] 
We assume that $\boldsymbol{\theta}$ is subject to these dynamic constraints:
\begin{equation}
    \label{eq: theta system}
    \boldsymbol{\theta}(t) = \boldsymbol{\theta}(t-1) + \omega(t), \,\,\,\,\, \omega(t) \sim \mathcal{N}(\mathbf{0}, \sigma_\omega^2 I).
\end{equation}
In other words, the parameter vector performs a random walk in the parameter space, with the step length controlled by the parameter $\sigma_\omega$. This hyperparameter, together with the order of the approximation $E$, corresponding to the number of representative points used in Sec. \ref{sec: Nystrom approximation}, are the only two parameters that DistKP requires beyond a standard Gaussian Process. Such a general assumption allows us to determine how smoothly the function varies over time. \newline
The measurements collected by each agent $i$ can now be considered as observations of the state of the system \eqref{eq: theta system}
\begin{equation}
    \label{eq: theta observation}
    y_i(t) = \Phi(\mathbf{x}_i(t))^\top \boldsymbol{\theta}(t) + \nu(t), \,\,\,\,\, \nu(t) \sim \mathcal{N}(0, \sigma_\nu^2).
\end{equation} 

By reformulating the time-varying estimation problem in this manner, it reduces to a state estimation problem for the system \eqref{eq: theta system}, \eqref{eq: theta observation}. For such linear system subject to Gaussian noise, it is known that a Kalman Filter \cite{doi:https://doi.org/10.1002/0470045345.ch5} provides the optimal linear recursive estimation method.
\subsection{Distributed regression}
We are interested in making the agents collaboratively learn a shared time varying Gaussian Process by communicating through the network described in Section \ref{sec: Problem formulation}.
In this scenario, a distributed implementation of the Kalman filter becomes necessary.
For this purpose, we apply the algorithm described in \cite{dist_kalman}, that keeps the local prediction and update step of a traditional Kalman filter and introduces an estimate sharing step. During this step different local estimates of $\boldsymbol{\hat{\theta}}$ are combined by applying a distributed version of classic measurement fusion techniques.

More specifically, consider the $n$ agent's local estimates $\boldsymbol{\hat{\theta}}_i$, organized in a vector
\begin{align*}
\boldsymbol{\hat{\theta}}(t) = \begin{bmatrix} \boldsymbol{\hat{\theta}}_1(t) \\ \boldsymbol{\hat{\theta}}_2(t) \\ \vdots \\ \boldsymbol{\hat{\theta}}_n(t) \end{bmatrix}. 
\end{align*}
The distributed Kalman Filter of \cite{dist_kalman} assumes that each agent has access to the local covariance matrices $P_i$, each referred to $\boldsymbol{\hat{\theta}}_i$,
but do not have access to each other's estimates and covariances directly. Using the communication network, we can employ an averaging algorithm to obtain at each agent: 
\[
    \frac{1}{n} \sum_{i=1}^n P_i^{-1} \quad \textrm{ and } \quad \frac{1}{n} \sum_{i=1}^n P_i^{-1} \boldsymbol{\hat{\theta}}_i.
\]

A convenient way to combine measurements, as shown in \cite{dist_kalman}, consists in neglecting the correlation between different local estimates and thus approximating the Best Linear Unbiased Estimator (BLUE) through inverse variance weighting.
However, averaging leads to highly correlated local estimates and, once the algorithm has converged to a shared covariance estimate $P$, each inverse variance weighting step introduces an error factor of $\frac{1}{n}$
\[
    P_i(t+1) \approx \left( \sum_{j=1}^n P_j(t)^{-1} \right)^{\!-1} = \frac{1}{n} P
\]
We correct this effect by introducing a multiplicative factor $n$ in the covariance formula.
\[
    \boldsymbol{\bar{\theta}} 
    \!=\! \left( \frac{1}{n} \sum_{i=1}^n P_i^{-1} \right)^{\!-1}\! \frac{1}{n} \sum_{i=1}^n P_i^{-1} \boldsymbol{\hat{\theta}}_i, \,\,\,
    \bar{P} = n \left( \sum_{i=1}^n P_i^{-1} \right)^{\!-1}\!\!.
\]
A comparison between the BLUE and our approximation is illustrated in Figure \ref{fig:inv_var_vs_BLUE}, that shows how our Kalman Filter introduces negligible errors. We point out that, without such approximation, it would not be possible to perform state estimation in a distributed fashion.

\begin{figure}
    \centering
    \includegraphics[width=0.4\textwidth]{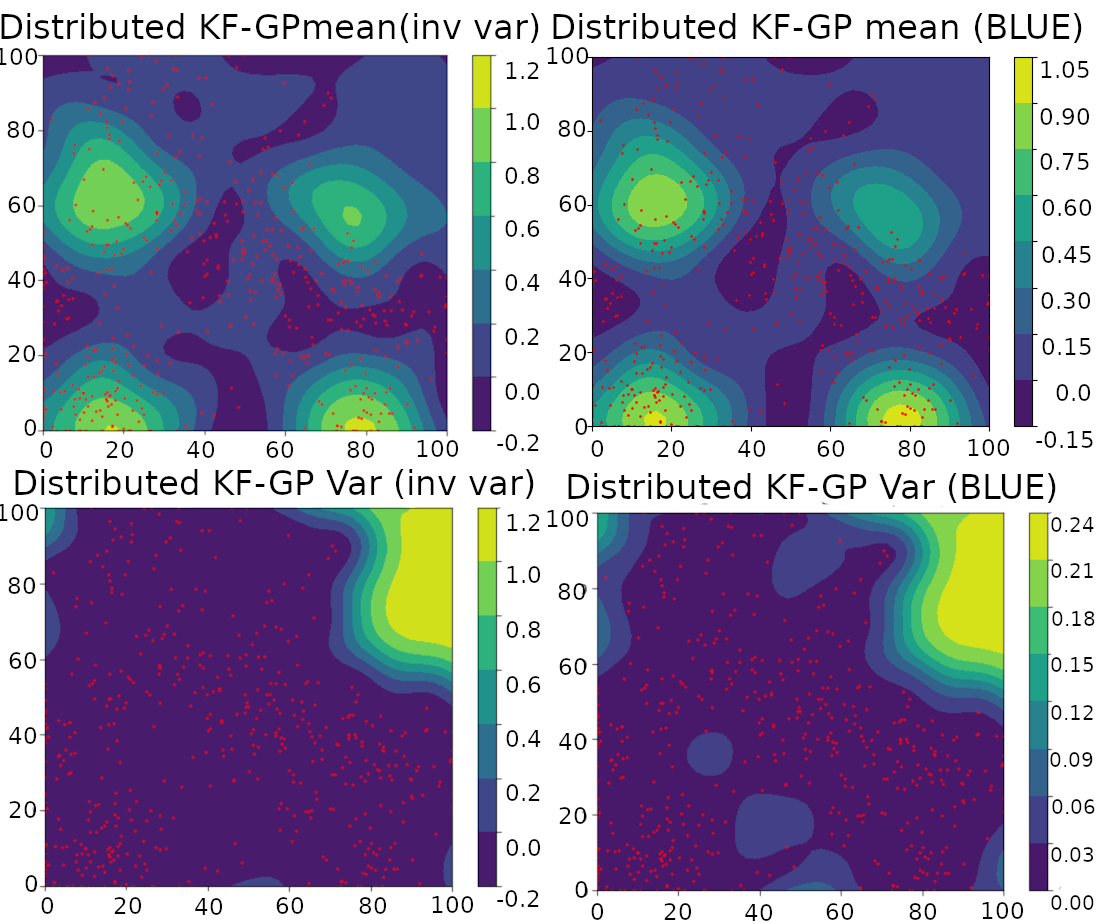}
    \caption{Comparison between measurement fusion with inverse variance weighting (left) and (centrally computed) BLUE (right). Sample locations are represented as red dots. We use $\sigma_\omega = 0$ since higher values make it difficult to compute the BLUE due to ill conditioned matrices. }
    \label{fig:inv_var_vs_BLUE}
\end{figure}

\section{DistKP Algorithm}
\label{sec: problem solution}
Building on the results presented so far, we introduce DistKP, as shown in Algorithm \ref{alg:one} and explained as follows:

\begin{algorithm}
\caption{DistKP for Agent $i$}
\begin{algorithmic}[1] 

    \State $P_i  \gets \sigma_{\text{init}}^2 \cdot I_E$  \Comment{Initialization}
    \State $\theta_i  \gets \texttt{zeros}((E,1))$
    \State $L_i \gets P_i^{-1} \theta_i$
    \State $\widetilde{P}_i \gets P_i^{-1}$
    \State $t \gets 0$
    \State $\mathcal{R} \gets \texttt{get\_repr\_points()}$ \Comment{Returns shared set}
    \State $K_{\mathcal{R} \mathcal{R}} \gets k(\mathcal{R}, \mathcal{R})$
    \State $U, \Lambda_{\mathcal{R}}, V \gets \texttt{svd}(K_{\mathcal{R} \mathcal{R}})$
    \State $\text{SensingPeriod} \gets 1$ \Comment{Choose appropriate value}
    \While{True}
        \Comment{Sensing} 
        \If{mod(t,\text{SensingPeriod}) = 0}
        \State $y_i = \texttt{sample\_function}(x_i)$
        \State $\theta_i \gets  \widetilde{P}_i^{-1}L$
        \State $P_i \gets \widetilde{P}_i^{-1}$
        \State $H \gets k(x_i, \mathcal{R}) \,\, U \Lambda_{\mathcal{R}}^{-\frac{1}{2}}$
        \State $K \gets (P_i + \sigma_\omega^2 I)  H^\top  (H  (P_i+\sigma_\omega^2 I)  H^\top + \sigma_\nu^2 I)^{-1}$
        \State $\theta_i  \gets \theta_i + K (y_i - H \theta_i)$  \Comment {Kalman update}
        \State $P_i = (I_E - K H) (P_i+\sigma_\omega^2 I)$
        \State $L_i \gets P_i^{-1} \theta_i$  \Comment{Initialize Consensus}
        \State $\widetilde{P}_i \gets P_i^{-1}$
        \EndIf
    \State \{ $L^s_j, \widetilde{P}^s_j : j \in \mathcal{N}_i \} = \text{incoming\_estimates()}$ 
    \Comment{Get estimates from neighbours}
    \State $L_i \gets \frac{1}{1 + |\mathcal{N}_i|} \left( L_i + \sum_{j\in\mathcal{N}_i} L^s_j \right)$. 
    \State $\widetilde{P}_i \gets \frac{1}{1 + |\mathcal{N}_i|} \left( \widetilde{P}_i +\sum_{j\in\mathcal{N}_i} \widetilde{P}^s_j \right)$
    \State $t \gets t+1$
    \EndWhile
\end{algorithmic}
\label{alg:one}
\end{algorithm}

Before online operation, we use the communication network to initialize each agent with the same set of representative points. For instance, these can be sampled uniformly at random from the domain $\mathcal{P}$ (as done in the simulations). These are then used to compute the matrices $U_{11}$ and $\Lambda_{\mathcal{R}}$ through singular value decomposition.
Each agent also requires knowledge of the kernel $k(\cdot, \cdot)$, and of the parameters $\sigma_\omega$, that represents the process noise of the function dynamics, and $\sigma_\nu$, the variance of the measurement noise.

During the online operations, each agent runs the remaining steps in which  they continuously measure $f$ and refine the local estimate through communication. 
\begin{enumerate}
    \item At each time step, each agent collects a noisy measurement of the function and then uses \eqref{eq:Nystrom Phi} to obtain the required features. Then a Kalman Filter based on model \eqref{eq: theta system}, \eqref{eq: theta observation} is used to update the local estimates $\boldsymbol{\hat{\theta}}_i(t)$ and $P_i$.
    \item The agents then proceed to do \texttt{SensingPeriod} steps of the consensus algorithm: Each agent computes the vector $P_i^{-1} \boldsymbol{\hat{\theta_i}}$ and communicates it to the neighbours, together with $P_i^{-1}$. Incoming variables are averaged with the local ones and $\boldsymbol{\hat{\theta_i}}$ and $P_i$ are updated. 
    As the averaging algorithm has an exponential convergence rate, precision increases with the number of iterations but, due to the dynamic nature of the function, the agents need to balance the time spent sampling and the number of consecutive consensus steps. This is done by tuning \texttt{SensingPeriod}. 
\end{enumerate}

\section{Numerical Results}
\label{sec: numerical results}
We present simulations to validate the effectiveness of DistKP, where we consider the case of a swarm of UAVs that need to learn the spatial distribution of thermal currents. This setting is strongly inspired by the Skyscanner project \footnote{\label{foot:skyscanner} https://sites.laas.fr/projects/skyscanner/} and by \cite{uav_sampling}. Thermal currents are wind currents occurring due to convection associated with cloud formations. Learning the distribution of these currents, that changes dynamically based on the cloud's position, is important in applications that exploit them to minimize the energy required to maintain altitude.

The unknown function is a simplified model of the wind distribution in a cumulus cloud, of which we consider a 2d horizontal section; We use a color map to render the wind direction (red is up and blue is down) and the wind power (represented by color intensity). The values range between -1 and +1, representing the minimum and maximum upwind value, respectively.
We assume a constant speed for the clouds, and we initialize a swarm of 16 UAVs by distributing them evenly spaced over the region. Each drone moves independently from the others according to a random walk with Gaussian steps. Our aim in this simulation is to demonstrate the adaptability of our learning algorithm. 

Figure \ref{fig:evolution} shows the evolution of the local mean and variance estimates for a single drone of the swarm. We adopt a Laplace kernel, that is well suited for modeling irregular functions like the wind field; this choice is enabled by the use of Nyström features and would not have been possible if we had adopted KL features as proposed in \cite{distributed_informative_planning}, since they cannot be computed in close-form for this kernel. After just 50 iterations (first row), the mean estimate closely resembles the ground truth. The algorithm also estimates the variance for each point of the map (right column), thus giving insights about the confidence of the prediction. 
At time step 600, for example, we observe an incorrect prediction around (10,10): the agent predicts strong down wind in an area with zero turbulence. But, if we look at the variance, we notice high uncertainty in the estimate. The variance can thus be used by a control algorithm to weigh information by its reliability and to guide exploration \cite{distributed_informative_planning}.
In general, the estimate follows the movement of the clouds, maintaining a useful representation for the entire period. 

Additional results can be seen in Figure \ref{fig:comparison}, where our algorithm (left) is compared with a distributed time-varying GP as described in \cite{distributed_informative_planning} (middle). In order to be able to use the Laplace kernel with the latter algorithm, we substituted KL features with Nyström ones in our implementation.
Because of the memory loss problem, the second algorithm estimate loses important information on the wind distribution over time, resulting in a deteriorated estimate.

Additional comparisons can be found in our Github repository\footnote{https://github.com/Nicola-Taddei/DistributedGP}, with all the code used for these simulations.
\begin{figure}
    \centering
    \includegraphics[width=1.0\linewidth]{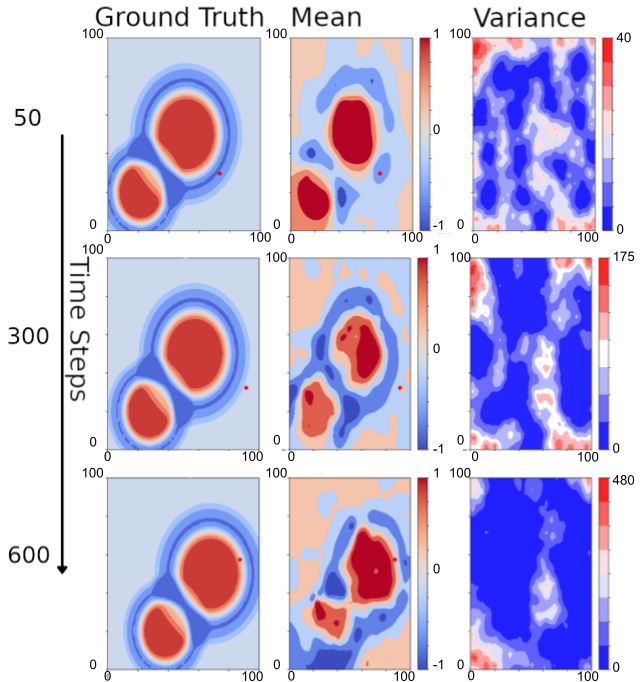}
    \caption{Ground Truth (left), mean (center) and variance (right) estimates at $t=50$ (up), $t=300$ (middle), $t=600$ (bottom).}
    \label{fig:evolution}
\end{figure}
\begin{figure}
    \centering
    \includegraphics[width=1\linewidth]{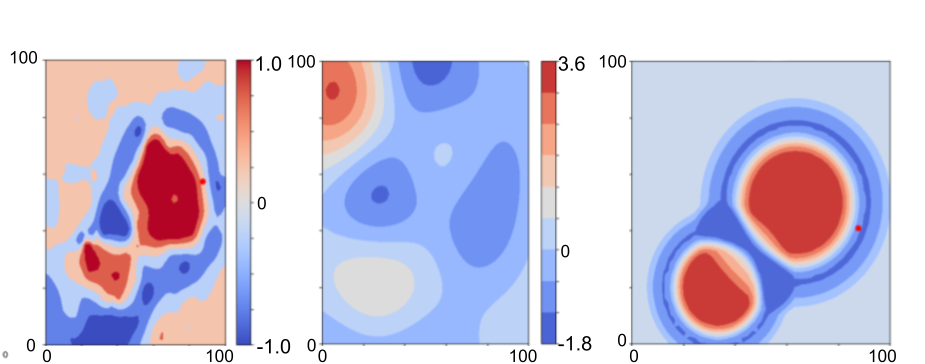}
    \caption{Comparison between our estimate at time $t=600$ (left), the estimate of a distributed GP with time-varying method from \cite{distributed_informative_planning} (center) and the ground truth (right).}
    \label{fig:comparison}
\end{figure}

\section{Conclusions}
In this paper, we introduced a new algorithm, DistKP, that enables distributed Gaussian Process regression in the time-varying setting by recasting it as an estimation problem. DistKP overcomes limitations of previous methods as it uses data more efficiently and relies on weaker assumptions on how the data is correlated across time. The results were validated numerically through extensive experiments on simulations inspired by the Skyscanner project \cite{uav_sampling}.

Future research will focus on providing convergence guarantees and error bounds for the proposed method as well as integrating it in a planning pipeline like in \cite{distributed_informative_planning}. 

\input{root.bbl}

\end{document}